\documentclass[11pt]{article}
\usepackage{amsmath,amsfonts,amssymb,amsthm}
\usepackage{fullpage}
\usepackage[ruled]{algorithm2e}

\usepackage{subfigure} \usepackage{epstopdf}
\usepackage{graphicx}
\usepackage{color}
\usepackage[sans]{dsfont}
\usepackage{mathtools}
\usepackage{hyperref}
\usepackage{cleveref}
\usepackage{tcolorbox}
\usepackage{multirow}
\usepackage{xspace}
\tcbuselibrary{skins,breakable}
\tcbset{enhanced jigsaw}

\newtheorem{theorem}{Theorem}

\newtheorem{corollary}[theorem]{Corollary}
\newtheorem{lemma}[theorem]{Lemma}
\newtheorem{observation}[theorem]{Observation}

\theoremstyle{definition}
\newtheorem{claim}[theorem]{Claim}

\newcommand{\set}[1]{\left\{ #1 \right\}}

\newcommand{\qset}{{\mathcal{Q}}}

\newcommand{\cset}{{\mathcal{C}}}

\newcommand{\cov}{{\mathsf{cov}}}

\newcommand{\dist}{\textnormal{\textsf{dist}}}
\newcommand{\gir}{\textnormal{\textsf{girth}}}

\def\floor#1{\left\lfloor #1 \right\rfloor}

\def\card#1{\left| #1 \right|}

\newcounter{note}
\newcommand{\znote}[1]{\refstepcounter{note}$\ll${\bf ZT~\thenote:}
	{\sf \color{red} #1}$\gg$\marginpar{\tiny\bf ZT~\thenote}}

\begin{document}

\begin{titlepage}
	
	\title{An $\tilde\Omega\big(\sqrt{\log |T|}\big)$ Lower Bound for Steiner Point Removal}
	
	\author{Yu Chen\thanks{EPFL, Lausanne, Switzerland. Email: {\tt yu.chen@epfl.ch}. Supported by ERC Starting Grant 759471.} \and Zihan Tan\thanks{Rutgers University, NJ, USA. Email: {\tt zihantan1993@gmail.com}. Supported by a grant to DIMACS from the Simons Foundation (820931).}} 
	
	\maketitle

	\thispagestyle{empty}
	\begin{abstract}
		
In the Steiner point removal (SPR) problem, we are given a (weighted) graph $G$ and a subset $T$ of its vertices called terminals, and the goal is to compute a (weighted) graph $H$ on $T$ that is a minor of $G$, such that the distance between every pair of terminals is preserved to within some small multiplicative factor, that is called the \emph{stretch} of $H$.

It has been shown that on general graphs we can achieve stretch $O(\log |T|)$ [Filtser, 2018]. On the other hand, the best-known stretch lower bound is $8$ [Chan-Xia-Konjevod-Richa, 2006], which holds even for trees.
In this work, we show an improved lower bound of $\tilde\Omega\big(\sqrt{\log |T|}\big)$.

	\end{abstract}
\end{titlepage}


\section{Introduction}

In the Steiner point removal (SPR) problem, we are given an undirected weighted graph $G$ and a subset $T$ of its vertices called \emph{terminals}, and the goal is to compute a minor $H$ of $G$ with $V(H)=T$ (and arbitrary edge weights), such that for every pair $t,t'\in T$,
$$\dist_G(t, t')\leq \dist_{H}(t, t')\leq \alpha\cdot \dist_G(t, t')$$ for some small real number $\alpha\ge 1$, that is called the \emph{stretch} of the SPR solution $H$.

The SPR problem was first proposed by Gupta \cite{gupta2001steiner}, who proved that trees admit SPR solutions with stretch at most $8$, which was later shown to be tight \cite{chan2006tight}. 
Later, Kamma, Nguyen, and Krauthgamer \cite{KNZ14} showed that on general graphs we can achieve stretch $O(\log^5 |T|)$, and this was later improved by Cheung \cite{cheung2018steiner} to $O(\log^2|T|)$ and further by Filtser \cite{filtser2018steiner} to $O(\log|T|)$. 

There has also been an exciting line of work on achieving $O(1)$-stretch for special types of graphs, especially minor-free graphs. 
Basu and Gupta \cite{basu2008steiner} showed that outerplanar graphs (graphs excluding $K_{2,3}$ and $K_4$ minor) admit $O(1)$-stretch SPR solutions. 
Filtser \cite{filtser2020scattering} reduced constructing SPR solutions to computing scattering partitions. Building upon this, Hershkowitz and Li \cite{hershkowitz20211} showed that  series-parallel graphs (graphs excluding $K_4$ minor) admit SPR solutions with $O(1)$ stretch, and in a very recent sequence of breakthrough papers, Chang et al \cite{chang2023covering,chang2023resolving,chang2023shortcut} finally showed that minor-free graphs admit $O(1)$-stretch SPR solutions.

However, after all this amazing progress on the upper bound side, the best-known stretch lower bound for general graphs remains $8$, which holds even for trees \cite{chan2006tight}. 
In this work, we prove the first superconstant lower bound for general graphs. Our main result can be summarized as the following theorem.

\begin{theorem}
\label{thm: main}
For any large enough integer $k>0$, there exists a graph $G$ and a subset $T\subseteq V(G)$ of $k$ terminals, such that any SPR solution $H$ of $G$ with respect to $T$ has stretch $\Omega(\sqrt{\log k/\log\log k})$.
\end{theorem}

Our lower bound is polynomially related to the best upper bound of $O(\log k)$ on general graphs. Closing this gap remains an interesting open problem.

\section{Preliminaries}

By default, all logarithms are to the base of $2$.

Let $G=(V,E)$ be an edge-weighted graph.
For a pair $v,v'$ of vertices in $G$, we denote by $\dist_{G}(v,v')$ the shortest-path distance between $v,v'$ in $G$.
Let $P$ be a path in $G$. 
We denote by $|P|$ its length (the total weight of all its edges).
The \emph{girth} of $G$, denoted by $\gir(G)$, is defined as the minimum length of any cycle in $G$. We use the following observation.

\begin{observation}
\label{obs: girth}
Let $P$ be a shortest path in $G$ connecting a pair $u,v$ of its vertices, and let $Q$ be another $u$-$v$ path (possibly non-simple) in $G$. Then either $P\subseteq Q$, or $|P|+|Q|\ge \gir(G)$.
\end{observation}
\begin{proof}
Assume that $P\not\subseteq Q$, and we will show that $P\cup Q$ contains a cycle, which implies \Cref{obs: girth}.
Let $e$ be an edge in $P\setminus Q$.
Note that the (multi-)graph $P\cup Q$ can be viewed as a non-simple cycle (from $u$ to $v$ along $P$, and then back to $u$ along the reverse of $Q$). We iteratively remove from it two copies of an edge that appear more than once in $P\cup Q$ until we cannot do so. It is easy to show that throughout the process the graph is always the union of several disjoint cycles, and in the end the graph consists of simple cycles and is non-empty (as edge $e$ will never get removed). Therefore, $P\cup Q$ contains a cycle.
\end{proof}


\subsubsection*{Minors and image paths.}

Let $G$ be a graph. We say that a graph $H$ is a \emph{minor} of $G$, iff $H$ can be obtained from $G$ by edge deletions, vertex deletions and edge contractions. We have the following immediate observation.

\begin{observation}
\label{obs: minor}
If $H$ is a minor of $G$, then $|E(H)|\le |E(G)|$.
\end{observation}

Recall that a SPR solution $H$ of $G$ is a minor of $G$ with $V(H)=T$, and $H$ can be edge-weighted.
Let $P=(t_0,\ldots,t_r)$ be a shortest path in $H$ connecting $t_0$ to $t_r$. We define its \emph{image path} in $G$ as follows. For each index $0\le i\le r-1$, we denote by $R_i$ the $t_i$-$t_{i+1}$ shortest path in $G$. We let $R$ be the sequential concatenation of paths $R_0,\ldots,R_{r-1}$, so $R$ is a path (possibly non-simple) in $G$ connecting $t_0$ to $t_r$. 
We use the following simple observations.

\begin{observation}
\label{obs: length}
The length of edge $(t_i,t_{i+1})$ in $H$ is at least $\dist_{G}(t_i,t_{i+1})$.
\end{observation}
\begin{proof}
This is simply because $\dist_{G}(t,t')\le \dist_{H}(t,t')$ is required for all pairs $t,t'\in T$.
\end{proof}

\begin{corollary}
\label{obs: minor dist}
$\dist_H(t_0,t_r)\ge |R|\ge \dist_G(t_0,t_r)$.
\end{corollary}
\begin{proof}
On the one hand, from \Cref{obs: length},
\[\dist_{H}(t_0,t_r)= \sum_{0\le i\le r-1}\dist_{H}(t_i,t_{i+1})\ge \sum_{0\le i\le r-1}\dist_{G}(t_i,t_{i+1})=\sum_{0\le i\le r-1}|R_i| = |R|.\]
On the other hand, $R$ is a path in $G$ connecting $t_0$ to $t_r$, so $|R|\ge \dist_G(t_0,t_r)$.
\end{proof}

\section{Proof of \Cref{thm: main}}

Recall that we are given a large enough integer $k$.
Throughout, we set $M = \floor{\sqrt{\log k\cdot \log\log k}}$. 

We construct the hard instance as follows.
Let $G'$ be a $3$-regular graph on $k$ vertices with girth at least $(\log k)/10$ 
(such graphs are known to exist, see e.g., \cite{alon2021high}).
We then add to $G'$, for each vertex $u\in V(G')$, a new vertex $t_u$ that is connected to $u$ by a single edge $(u,t_u)$. The resulting graph is denoted by $G$. 
Each edge in $E(G')$ has weight $1$, and each new edge $(u,t_u)$ has weight $M$. The terminal set is defined to be $T=\set{t_u\mid u\in V(G')}$\footnote{We remark that our construction is similar to the ``expander with tails'' graph used in \cite{calinescu2005approximation} for analyzing the integrality gap of the semi-metric relaxation LP for the $0$-Extension problem. But our analysis here is quite different from theirs.}.


Instead of directly proving \Cref{thm: main}, we will in fact prove the following stronger lemma, which, together with \Cref{obs: minor} and \Cref{obs: length}, immediately imply that any SPR solution of the instance $(G,T)$ has stretch $\Omega(\sqrt{\log k/\log\log k})$, thereby proving \Cref{thm: main}.

\begin{lemma}
\label{lem: main}
Let $H$ be a graph on $T$ such that (i) $|E(H)|\le |E(G)|$; and (ii) every edge $(t,t')$ in $H$ has length at least $\dist_{G}(t,t')$.
Then there exists a pair $t,t'\in T$, such that \[\dist_{H}(t,t')\ge \Omega(\sqrt{\log k/\log\log k})\cdot \dist_{G}(t,t').\]
\end{lemma}

From now on we focus on proving \Cref{lem: main}.

We define $\qset$ to be the set of paths $Q$ in $G$, such that 
\begin{itemize}
\item $Q$ is the image path of a single edge $(t_u,t_v)$ in $H$; and
\item $\dist_{G}(t_u,t_v)\le \log k/100$.
\end{itemize} 
We say a path $P$ in $G$ is \emph{covered} by a collection $\qset'\subseteq \qset$ of paths, iff $P\subseteq \bigcup_{Q\in \qset'}Q$. 
For such a path $P$, we denote by $\cov(P)$ be the minimum number of paths we need in $\qset$ to cover $P$. That is,
$\cov(P)=\min\set{|\qset'|\mid \qset'\subseteq \qset, P \text{ covered by }\qset}$.
We prove the following lemma.

\begin{lemma} \label{lem:cover}
There is a path $P$ in $G'$ of length $M$ such that $\cov(P)=\Omega(M/\log \log k)$.
\end{lemma}

\begin{proof}
Recall that $G'$ is the subgraph of $G$ without terminals.
%
We will construct a random path $P$, and then show that $\mathbb{E}[\cov(P)]=\Omega(M/\log \log k)$, which implies the lemma.

We denote $P=(u_1,\ldots,u_M)$, so each vertex $u_i$ of $P$ is a random variable. We start by pick a vertex in $G'$ uniformly at random and designate it as $u_1$, and we choose a random neighbor of this vertex in $G'$ as designate it as $u_2$. Then sequentially for each $i=2,3,\ldots,M-1$, we choose a random non-$u_{i-1}$ neighbor of $u_i$ in $G'$, and designate it as $u_{i+1}$. This completes the construction of $P$.

Denote $S=\floor{10 \log \log k}$, we partition the path $P$ into $\floor{M/S}$ subpaths $P_1,\ldots,P_{\floor{M/S}}$, where for each $1\le j< \floor{M/S}$, $P_j=(u_{(j-1)S+1},\ldots,u_{jS+1})$, and $P_{\floor{M/S}}=(u_{(\floor{M/S}-1)S+1},\ldots,u_{M})$, so each subpath $P_j$ has length at least $S$. We use the following two observations.

\begin{observation}
\label{obs: cov}
$\cov(P)\ge \sum_{1\le j\le \floor{M/S}}\big(\cov(P_j)-1\big)$.
\end{observation}
\begin{proof}
Consider a path $Q\in \qset$.
Note that both $P$ and $Q$ are simple. As $P$ contains $M$ edges and $Q$ contains at most $\log k/100$ edges, and $M+\log k/100< \gir(G')$, $P\cup Q$ may not contain any cycle, so $Q\cap P$ must be a continuous subpath of $P$.

Consider the optimal covering of $P$ in a sequence: $(Q_1,\ldots,Q_r)$, where paths are indexed according to the increasing order of the earliest intersection with $P$. Now for each $P_j$, we take the minimal subsequence $(Q_p,\ldots,Q_q)$ that covers $P_j$. It is easy to verify that the subsequences for $P_j$ and $P_{j+1}$ share at most one path $Q_q$. The observation now follows.
\end{proof}

\begin{observation}
The number of length-$S$ paths in $G'$ is $k\cdot 3\cdot 2^{S-1}$, and
for each $1\le j< \floor{M/S}$, path $P_j$ is uniformly at random chosen from them.
\end{observation}
\begin{proof}
For choosing a length-$S$ path $P$ (as our algorithm above), the first vertex has $k$ choices, the second has $3$ choices, and all the rest have $2$ choices. As $S<\gir(G)$, we cannot repeat any vertex and so will always get a simple path, and therefore the number of length-$S$ paths in $G'$ is $k\cdot 3\cdot 2^{S-1}$.

We index vertices of $G'$ by integers $1,\ldots,k$.
Consider a vertex $x$ in $G'$ whose neighbors are indexed by $y_1,y_2,y_3$, where $y_1\le y_2 \le y_3$. For each $\ell\in \set{1,2,3}$, we say that the \emph{upper sibling} of $y_\ell$ with respect to $x$ is $y_{\ell+1}$, and \emph{lower sibling} of $y_\ell$ with respect to $x$ is $y_{\ell-1}$ (here we use the convention that $3+1=1$).
From the above discussion, the path $P$ is determined by a sequence $(a,b,c_3,\ldots,c_{M})$, where $a,b\in \set{1,\ldots,k}$, representing the indices of the first and the second vertices; and for each $r\ge 3$, $c_r\in \set{\uparrow,\downarrow}$, reflecting that the $u_{r}$ is either the upper sibling of $u_{r-2}$ with respect to $u_{r-1}$ ($\uparrow$) or the lower sibling of $u_r$ with respect to $u_{r+1}$ ($\downarrow$).

As $M<\gir(G)/2$ and $G'$ is $3$-regular, the $M$-neighborhood of each vertex in $G'$ is isomorphic (to the graph where a vertex is connected to the roots of three complete binary trees). Thus, for each length-$S$ path $Q$ in $G'$, and for each $i$, the number of prefixes $(a,b,c_3,\ldots,c_{i+S})$ determining that $(u_i,\ldots,u_{i+S})=Q$ is exactly $2^{i-1}$ (the choices of $c_{i+2},\ldots,c_{i+S}$ are uniquely determined by the path $Q$, and reversely for each $c_{i+1},c_i,\ldots,c_3$, there are two choices, and in the end indices $a,b$ are uniquely determined). As the sequence $(a,b,c_3,\ldots)$ is generated uniformly at random, for each $i$, the random subpath $(u_i,\ldots,u_{i+S})$ of $P$ is equally likely to be realized as any length-$S$ path in $G'$. The observation now follows.
\end{proof}

\begin{observation}
A path in $\qset$ may cover at most $\log k/100$ length-$S$ paths in $G'$.
\end{observation}
\begin{proof}
Denote a path in $\qset$ as $(e_1,e_2,\ldots,e_r)$ with $r\le \log k/100$. If this path covers a length-$S$ path $L$, as this path is a simple path, then $L$ must be $(e_i,\ldots,e_{i+S})$ for some $1\le r\le L$, so there can be at most $r\le \log k/100$ such paths $L$.
\end{proof}

Note that $|\qset|\le |E(H)|\le |E(G')|=4k$, so the number of length-$S$ paths in $G'$ that can be covered by some path in $\qset$ is at most $4k\cdot (\log k/100)$.
Therefore, for each $1\le j< \floor{M/S}$, 
\[\Pr[\cov(P_j)\ge 2]=1-\Pr[\cov(P_j)= 1]\ge 1-\frac{4k\cdot (\log k/100)}{k\cdot 3\cdot 2^{S-1}}>0.9,\] 
as $2^{S-1}>\log k$.
Therefore, for each $1\le j< \floor{M/S}$, $\mathbb{E}[\cov(P_j)-1]\ge 0.9$, and so by \Cref{obs: cov}, 
$\mathbb{E}[\cov(P)]\ge 0.9\cdot (\floor{M/S}-1)=\Omega(M/\log\log k)$.
This completes the proof of the lemma.
\end{proof}

Let $u,v$ be the endpoints of the path $P$ given by \Cref{lem:cover}. We augment $P$ with edge $(u,t_u)$ at its beginning and edge $(v,t_v)$ at its end, to obtain a path, which we call $P'$ that connects $t_u$ to $t_v$.
As $|P'|=2M+M<\gir(G)/2$, $P'$ is the shortest path in $G$ connecting $t_u$ to $t_v$, so $\dist_{G}(t_u,t_v)=3M$.

Consider now the shortest $t_u$-$t_v$ path $Q$ in graph $H$ and its image path $R$ in $G$. 
If $P\not\subseteq R$, then from \Cref{obs: girth}, $|R|\ge \gir(G)-|P|\ge \log k/100$.
If $P\subseteq R$, then either $Q$ contains an edge whose length is at least $\log k/100$ which immediately implies $|R|\ge \log k/100$, or $Q$ only contains edges with length at most $\log k/100$, in which case from \Cref{lem:cover} we know that $|E(Q)|=\Omega(M/\log \log k)$, which, combined with the fact that the distance in $G$ between every pair of terminals is at least $2M$, implies that $|R|\ge \Omega(M\cdot (M/\log \log k))$. 
Finally, from \Cref{obs: minor dist} and $M = \floor{\sqrt{\log k\cdot \log\log k}}$,
\[
\frac{\dist_{H}(t_{u},t_{v})}{\dist_{G}(t_{u},t_{v})}\ge \frac{\min\set{\log k/100, \Omega(M^2/\log \log k) }}{3M}=\Omega(\sqrt{\log k/\log\log k}).\]

\vspace{-10pt}
\paragraph{Acknowledgement.} We would like to thank D Ellis Hershkowitz for bringing this problem to our attention. We would also like to thank the organizers of the DIMACS Workshop on Modern Techniques in Graph Algorithms at Rutgers in Summer 2023.

\bibliographystyle{alpha}
\bibliography{REF}

\end{document}